\documentclass[11pt]{article}

\usepackage[utf8]{inputenc}
\usepackage[T1]{fontenc}
\usepackage{lmodern}

\usepackage{amsmath, amssymb, amsthm}
\usepackage{geometry}
\usepackage{tikz}
\usetikzlibrary{arrows.meta,calc}

\usepackage{hyperref}
\hypersetup{
  colorlinks=true,
  linkcolor=blue,
  citecolor=blue,
  urlcolor=blue
}

\newtheorem{theorem}{Theorem}[section]
\newtheorem{proposition}[theorem]{Proposition}

\theoremstyle{remark}
\newtheorem*{remark}{Remark}

\title{Curvature--Driven Dynamics on $S^3$: A Geometric Atlas}
\author{Evgeny A.~Mityushov}
\date{}

\begin{document}

\maketitle

\begin{abstract}
We develop a geometric atlas of dynamical regimes on the rotation group
$SU(2)$, combining geodesic flows, heavy rigid body dynamics, and a
curvature--based decomposition of the Euler--Poisson equations.
Representing the equations of motion in the form
\[
\dot{\Omega} = K_{\mathrm{geo}}(\Omega)
+ K_{\mathrm{ext}}(\Gamma),\qquad
\dot{\Gamma} = \Gamma \times \Omega,
\]
we interpret rigid--body motion as the interaction of inertial and
external curvature fields.
This unified viewpoint recovers classical integrable cases
(Lagrange, Kovalevskaya, Goryachev--Chaplygin) from a single geometric
mechanism and clarifies their geodesic prototypes on $SU(2)$.

The central new result is the identification and geometric explanation of a pure--precession family in the inertia ratio $(2,2,1)$,
obtained from a curvature--balanced geodesic regime with the same inertia ratio.
The corresponding pure--precession regime for the $(2,2,1)$ heavy top was first identified in \cite{Mityushov221}; here we place it into a curvature--based atlas and interpret it as a balance between inertial and external curvature fields.
We also exhibit a schematic curvature diagram organizing the main dynamical regimes.
Finally, we outline GCCT (Geometric Curvature Control Theory), a
curvature--driven approach to control on $S^3$ designed to produce smooth
globally regular controls suitable for benchmark maneuvers; a detailed
comparison with Pontryagin--type optimal solutions is left for future work.
\end{abstract}

\tableofcontents

\section{Introduction}

The dynamics of a rigid body on $SO(3)$ or $SU(2)$ has long served as a
central testing ground for geometric ideas in mechanics. Classical
integrable systems---those of Lagrange, Kovalevskaya, and
Goryachev--Chaplygin---anticipated later developments in modern
differential geometry, Lie theory, and Hamiltonian mechanics; see, for
example, Arnold~\cite{Arnold}, Bolsinov--Fomenko~\cite{BolsinovFomenko},
and the classical papers of Kovalevskaya and Chaplygin
\cite{Kovalevskaya,Chaplygin}. On the geometric side, the behavior of
left--invariant metrics on Lie groups and their curvature has been
extensively studied, beginning with Milnor~\cite{Milnor}.

From the classical geometric viewpoint, such systems are often organized either
via invariants of the inertia tensor and associated first integrals, or via
symplectic and Poisson reduction frameworks such as the Euler--Poincar\'e
equations. These descriptions are extremely powerful for identifying integrable
cases and understanding the global structure of the phase portrait, but they
do not always make it transparent which geometric features of the inertia
tensor are directly responsible for specific families of trajectories.

In this paper we instead use the curvature of the left--invariant metric
induced by the inertia tensor as the primary organizing object. At the level
of metric Lie algebras, this curvature coincides with the Riemann tensor of
the Levi--Civita connection studied by Milnor, so our construction is fully
compatible with his classification of possible curvature signatures. The new
ingredient is that we further decompose the curvature into dynamically
meaningful subspaces and use this decomposition to partition the space of
inertia tensors into regimes. In this way the Geometric Atlas we obtain
refines Milnor's picture by indicating which curvature components generate
pure--precession type motions, which support classical integrable tops, and
which lead to more generic mixed behavior.

In this work we construct a unified Geometric Atlas of dynamical regimes
on $SU(2)$ based on a curvature decomposition of the Euler--Poisson
equations:
\[
\dot{\Omega} = K_{\mathrm{geo}}(\Omega) + K_{\mathrm{ext}}(\Gamma),
\qquad
\dot{\Gamma} = \Gamma \times \Omega.
\]
Here $K_{\mathrm{geo}}$ is the inertial curvature field determined by the
left--invariant metric on $SU(2)$, while $K_{\mathrm{ext}}$
encodes the external field in the body frame. The Atlas is organized
into a finite list of curvature regimes, each associated with a simple
normal form for $(K_{\mathrm{geo}}, K_{\mathrm{ext}})$. Within each regime we identify
canonical geodesic patterns and their heavy--top counterparts, including
classical integrable systems and a new curvature--balanced regime with
inertia ratio $(2,2,1)$ that produces pure precession. A schematic
projection of these regimes into the $(I_2/I_1,I_3/I_1)$--plane is shown in
Figure~\ref{fig:curvature-atlas}.

\begin{figure}[h]
\centering
\begin{tikzpicture}[scale=1.1]
  \draw[-{Latex}] (-0.2,0) -- (5.2,0) node[below] {$I_2/I_1$};
  \draw[-{Latex}] (0,-0.2) -- (0,3.2) node[left] {$I_3/I_1$};

  \fill (1,1) circle (1.5pt) node[above right] {Spherical};

  \draw[dashed] (1,0.2) -- (1,3);
  \node[above] at (1,3) {Lagrange};

  \draw[dotted] (0.3,0.5) -- (4.8,0.5);
  \node[right] at (4.8,0.5) {Kovalevskaya};

  \draw[dotted] (0.3,0.25) -- (4.8,0.25);
  \node[right] at (4.8,0.25) {Goryachev--Chaplygin};

  \fill (1,0.5) circle (1.5pt) node[above left] {$(2,2,1)$};

  \node at (3,2) {generic anisotropic};

\end{tikzpicture}
\caption{Schematic curvature atlas in the $(I_2/I_1,I_3/I_1)$--plane.}
\label{fig:curvature-atlas}
\end{figure}
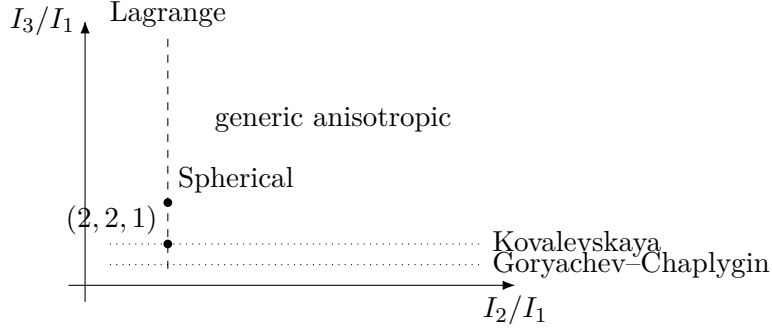

\section{Geodesic Dynamics on $SU(2)$}
\subsection{Left--Invariant Metric and Curvature Field}

A left--invariant metric on $SU(2)$ is determined by the inertia tensor
$I = \mathrm{diag}(I_1, I_2, I_3)$. The geodesic equation takes the
Euler--Poincar\'e form:
\[
\dot{\Omega} = I^{-1}(I\Omega \times \Omega).
\]
We rewrite this as the inertial curvature field
\[
\dot{\Omega} = K_{\mathrm{geo}}(\Omega),
\]
which is the organizing object of the Atlas.

\begin{remark}[Geometric meaning of $K_{\mathrm{geo}}$]
Formally, $K_{\mathrm{geo}}$ is the geodesic spray of the left--invariant
metric on $SU(2)$ defined by the inertia tensor $I$, written in body
coordinates. Equivalently, it is obtained from the Levi--Civita connection
of this metric by expressing the geodesic equation
$\nabla_{\dot q}\dot q = 0$ in the body frame. In this sense $K_{\mathrm{geo}}$
is literally a curvature object in the differential--geometric sense, and
the inertial curvature space we study below is the finite--dimensional space
spanned by its components in a fixed body basis.
\end{remark}

Throughout Section~2 we work with $K_{\mathrm{geo}}$ modulo a positive
time--rescaling factor (i.e., up to multiplication by a nonzero
constant), since geodesic trajectories on $SU(2)$ are unchanged by
uniform reparametrization. Coefficients are therefore normalized for
clarity and need not match the dynamical scaling used in Section~3. In
particular, the numerical coefficients that appear in the explicit
formulas for $K_{\mathrm{geo}}$ in Sections~2.5--2.6 are written in this
normalized form and should not be interpreted as the original physical
inertia parameters once time has been rescaled.

\subsection{A1. Spherical Metric}

If $I_1 = I_2 = I_3$, then $K_{\mathrm{geo}}\equiv 0$ and all geodesics are
one--parameter subgroups:
\[
\Omega(t) = \Omega_0 = \text{const},\qquad
q(t) = q(0)\,\exp\!\left(\frac{t}{2}\,\Omega_0\right).
\]
This yields great circles on $S^3$ and serves as the isotropic reference
point in the classification.

\subsection{A2. Lagrange Metric}

For $I_1 = I_2 \neq I_3$, the curvature field has the form
\[
K_{\mathrm{geo}} =
\left(
\frac{I_1 - I_3}{I_1}\omega_2\omega_3,\;
\frac{I_3 - I_1}{I_1}\omega_1\omega_3,\;
0
\right).
\]
The third component vanishes identically, and the remaining two lie in
the $(\omega_1,\omega_2)$--plane. This corresponds to the Lagrange top
in the presence of an external field aligned with the distinguished axis.

\subsection{A3. Generic Anisotropic Metric}

For pairwise distinct inertia moments $I_1,I_2,I_3$, the curvature field
has three nontrivial components and no special cancellations. The geodesic
flow exhibits mixed curvature features and lacks the simple symmetries of
the Lagrange or Kovalevskaya cases.

\subsection{A4. Kovalevskaya Metric}

If $I_1 = I_2 = 2I_3$, then the curvature field simplifies to
\[
K_{\mathrm{geo}} = (\omega_2\omega_3,\; -\omega_1\omega_3,\; 0).
\]
Two components form a symmetric pair while the third vanishes. This
structure underlies the classical Kovalevskaya case in the presence of
an external field and appears here as a distinguished geodesic pattern
on $SU(2)$; see~\cite{Kovalevskaya,BolsinovFomenko}.

\subsection{A5. Goryachev--Chaplygin Metric}

If $I_1 = I_2 = 4I_3$, then
\[
K_{\mathrm{geo}} = (2\omega_2\omega_3,\; -2\omega_1\omega_3,\; 0),
\]
representing a scaled version of the Kovalevskaya pattern with
stronger anisotropy. The geodesic flow has more pronounced oscillatory
behavior in the $(\omega_1,\omega_2)$--plane and is associated with the
Goryachev--Chaplygin top when an external field is added; see, e.g.,
\cite{Chaplygin,BolsinovFomenko}.

\subsection{A6. The Curvature--Balanced Regime $(2,2,1)$}

We denote by $(2,2,1)$ the normalized \emph{inertia ratio}
$I_1:I_2:I_3 = 2:2:1$, i.e.\ after rescaling of time and energy so that
the principal moments satisfy $I_1=I_2=2$ and $I_3=1$.
In this case
\[
K_{\mathrm{geo}} =
\left(
\frac12\omega_2\omega_3,\;
-\frac12\omega_1\omega_3,\;
0
\right).
\]
This minimal anisotropic structure produces elementary closed geodesics
with particularly simple curvature behavior. The dynamical significance of this inertia ratio for the heavy top, including the existence of a pure--precession regime, was first analyzed in \cite{Mityushov221}.
It will play a central role here as a curvature--balanced regime yielding pure precession in the Euler--Poisson setting.

\subsection{A7. Summary of Curvature Regimes}

For later reference we summarize the main curvature regimes appearing in
this section.
\begin{center}
\begin{tabular}{lll}
\hline
Regime & Inertia pattern & Curvature feature \\
\hline
Spherical & $I_1=I_2=I_3$ & isotropic, $K_{\mathrm{geo}}\equiv 0$ \\
Lagrange & $I_1=I_2\neq I_3$ & orthogonal splitting of curvature fields \\
Kovalevskaya & $I_1=I_2=2I_3$ & symmetric pair in $(\omega_1,\omega_2)$--plane \\
Goryachev--Chaplygin & $I_1=I_2=4I_3$ & stronger anisotropy of Kovalevskaya type \\
$(2,2,1)$ & $I_1:I_2:I_3 = 2:2:1$ & curvature--balanced pure--precession regime \\
\hline
\end{tabular}
\end{center}

\section{Heavy Rigid Body Dynamics: Curvature Formulation}

In the presence of an external field, the rigid body dynamics is
described by the Euler--Poisson equations
\[
I\dot{\Omega} = I\Omega\times\Omega + \mu\times\Gamma,\qquad
\dot{\Gamma} = \Gamma\times\Omega,
\]
where $\Gamma$ is the direction of the field in the body frame and
$\mu$ is the radius vector of the center of mass. We now rewrite these
equations in curvature form and classify heavy tops by the interaction
of inertial and external curvature.

In this section we work with the \emph{physical} principal moments
$I_1,I_2,I_3$ rather than the normalized coefficients used in Section~2.

\subsection{B1. Curvature Decomposition of the Euler--Poisson Equations}

Define the inertial curvature field by
\[
K_{\mathrm{geo}}(\Omega) =
\left(
\frac{I_2-I_3}{I_1}\omega_2\omega_3,\;
\frac{I_3-I_1}{I_2}\omega_1\omega_3,\;
\frac{I_1-I_2}{I_3}\omega_1\omega_2
\right),
\]
and let the external curvature field be
\[
K_{\mathrm{ext}}(\Gamma) = I^{-1}(\mu\times\Gamma).
\]
Then the Euler--Poisson system can be written as
\[
\dot{\Omega} = K_{\mathrm{geo}}(\Omega)
+ K_{\mathrm{ext}}(\Gamma),\qquad
\dot{\Gamma} = \Gamma\times\Omega.
\]
This form makes transparent the curvature balance conditions underlying
various integrable and near--integrable regimes.

\subsection{B2. Lagrange Top as Orthogonally Balanced Curvature}

For the Lagrange top, $I_1 = I_2$, $\mu$ is aligned with the symmetry
axis, and $K_{\mathrm{ext}}$ has only one nonzero component. In the curvature
picture, $K_{\mathrm{geo}}$ lies in the $(\omega_1,\omega_2)$--plane and
$K_{\mathrm{ext}}$ is purely in the $3$--direction. The two curvature fields are
orthogonal as elements of the curvature space, with respect to the inner
product induced by the kinetic energy; in particular this is an algebraic
orthogonality of curvature subspaces rather than a statement about spatial
right angles between physical axes. This orthogonality underlies the
integrability of the Lagrange top: one can separate the evolution of
$\omega_3$ and reduce the remaining dynamics to quadratures. In this sense
the Lagrange case is an orthogonally balanced curvature regime;
see~\cite{Arnold,BolsinovFomenko}.

\subsection{B3. Curvature--Balanced Regime $(2,2,1)$ and Pure Precession}

We now focus on the mixed anisotropic regime with inertia ratio
$I_1:I_2:I_3 = 2:2:1$.
In this case the inertial curvature field has, after normalization,
the form
\[
K_{\mathrm{geo}} =
\left(
\frac12\omega_2\omega_3,\;
-\frac12\omega_1\omega_3,\;
0
\right),
\]
while the external curvature field depends linearly on $\Gamma$. We
consider curvature balance conditions of the form
\[
K_{\mathrm{geo}}(\Omega_0) + K_{\mathrm{ext}}(\Gamma_0) = 0,\qquad
\|\Gamma_0\| = 1,
\]
which single out special initial conditions leading to pure precession.

\begin{proposition}
Let the inertia ratio satisfy $I_1:I_2:I_3 = 2:2:1$ and let $\mu = (1,0,0)$.
Then the Euler--Poisson system admits nontrivial solutions with
$\Omega(t) \equiv \Omega_0$. They
are exactly the initial conditions satisfying
\[
K_{\mathrm{geo}}(\Omega_0) + K_{\mathrm{ext}}(\Gamma_0) = 0,\qquad
\|\Gamma_0\| = 1.
\]
Along these trajectories $\Gamma(t)$ undergoes uniform precession around
the fixed axis $\Omega_0$.
\end{proposition}

\begin{proof}[Proof sketch]
If $\Omega(t)\equiv\Omega_0$, then $\dot{\Omega}=0$ and the
Euler--Poisson equations reduce to
\[
I\Omega_0\times\Omega_0 + \mu\times\Gamma_0 = 0,
\]
which is exactly the curvature balance condition
$K_{\mathrm{geo}}(\Omega_0)+K_{\mathrm{ext}}(\Gamma_0)=0$ with
$\|\Gamma_0\|=1$.
Conversely, any initial condition satisfying this balance makes the right
hand side of the $\dot{\Omega}$--equation vanish, so $\Omega(t)\equiv\Omega_0$.
The remaining equation $\dot{\Gamma}=\Gamma\times\Omega_0$ has solutions
given by rotation of $\Gamma$ at constant angular speed around $\Omega_0$,
i.e.\ uniform precession at a constant angle to $\Omega_0$.
\end{proof}

\begin{remark}
The proposition singles out a curvature--balanced family of pure--precession
trajectories inside the phase space of the Euler--Poisson system in the
$(2,2,1)$ case; it does not assert algebraic integrability of the full
dynamics. Away from this family one observes the usual mixed quasi--periodic
behavior typical for nonintegrable heavy tops. The existence of such pure--precession solutions for the $(2,2,1)$ heavy top was established in \cite{Mityushov221}; the curvature formulation given here shows that they arise from an exact balance of inertial and external curvature.
\end{remark}

We look for regimes with constant angular velocity $\Omega_0$, so that
$\dot{\Omega} = 0$. The curvature balance condition
\[
K_{\mathrm{geo}}(\Omega_0) + K_{\mathrm{ext}}(\Gamma_0) = 0
\]
then implies that the trajectory remains in a pure precession cone, with
$\Omega(t)\equiv\Omega_0$ and $\Gamma(t)$ tracing a circle at a constant
angle to $\Omega_0$.

\section{Geometric Applications and Curvature--Driven Control}

In this section we indicate how the curvature viewpoint leads to a
natural control design framework, which we call GCCT (Geometric
Curvature Control Theory).

\subsection{C1. Controlled Euler--Poisson System}

We add a control torque $u(t)\in\mathbb{R}^3$ in body coordinates:
\[
I\dot{\Omega} = I\Omega\times\Omega + \mu\times\Gamma + u,\qquad
\dot{\Gamma} = \Gamma\times\Omega.
\]
In terms of curvature fields this is
\[
\dot{\Omega} = K_{\mathrm{geo}}(\Omega)
+ K_{\mathrm{ext}}(\Gamma)
+ K_{\mathrm{c}}(u),\qquad
K_{\mathrm{c}}(u) := I^{-1}u.
\]
GCCT is based on decomposing $u$ according to the curvature--induced
splitting of the state space: roughly speaking, we write
\[
u = u_{\mathrm{pre}} + u_{\mathrm{mix}},
\]
where $u_{\mathrm{pre}}$ acts inside a ``pure--precession'' curvature
subspace, while $u_{\mathrm{mix}}$ transfers energy between curvature
subspaces.

\subsection{C2. Curvature Coordinates in the $(2,2,1)$ Regime}

In the $(2,2,1)$ case, the inertial curvature field $K_{\mathrm{geo}}$
splits naturally into an ``equatorial'' part in the $(\omega_1,\omega_2)$
plane and an ``axial'' direction determined by $\omega_3$. This suggests
introducing curvature coordinates $(\kappa_{\mathrm{pre}},\kappa_{\mathrm{mix}})$
on the set of admissible controls by imposing
\[
K_{\mathrm{c}}(u_{\mathrm{pre}}) \in
\mathrm{span}\{K_{\mathrm{geo}}(\Omega_0)\},\qquad
K_{\mathrm{c}}(u_{\mathrm{mix}}) \perp
\mathrm{span}\{K_{\mathrm{geo}}(\Omega_0)\}
\]
with respect to the kinetic energy inner product.
In these coordinates,
\begin{itemize}
  \item $u_{\mathrm{pre}}$ preserves the curvature balance
  $K_{\mathrm{geo}}(\Omega) + K_{\mathrm{ext}}(\Gamma) \approx 0$ and thus
  keeps the trajectory on (or close to) the pure--precession family;
  \item $u_{\mathrm{mix}}$ rotates the curvature balance cone, changing the
  orientation of the precession axis without strongly exciting transversal
  oscillations.
\end{itemize}

\subsection{C3. A Simple Benchmark Maneuver}

As a concrete benchmark, fix the inertia ratio $I_1:I_2:I_3 = 2:2:1$,
take $\mu=(1,0,0)$ and consider the following problem:
\begin{quote}
Rotate the precession cone by a prescribed angle $\delta$ around a
target axis while keeping the precession angle and kinetic energy
approximately constant.
\end{quote}
We choose an initial condition $(\Omega_0,\Gamma_0)$ on the
curvature--balanced family, so that $K_{\mathrm{geo}}(\Omega_0)
+ K_{\mathrm{ext}}(\Gamma_0) = 0$ and $\Omega_0$ is the precession axis.

A simple GCCT steering law is obtained as follows.
\begin{enumerate}
  \item On a short time interval $[0,T]$ we set $u_{\mathrm{pre}}(t)$ so
  that $K_{\mathrm{geo}}(\Omega(t)) + K_{\mathrm{ext}}(\Gamma(t))$ stays
  close to zero, i.e.\ we correct small deviations from the curvature
  balance produced by the mixing control.

  \item We choose $u_{\mathrm{mix}}(t)$ to be approximately constant and
  orthogonal (in curvature space) to $K_{\mathrm{geo}}(\Omega_0)$, with
  magnitude tuned so that the resulting rotation of the precession axis
  over time $T$ equals $\delta$.
\end{enumerate}
This yields a smooth, globally regular control $u(t)$ that carries
$(\Omega,\Gamma)$ from one point of the pure--precession cone to
another. We expect such curvature--based steering laws to be competitive
with classical optimal controls for moderate values of $\delta$; a
systematic numerical comparison with Pontryagin--type extremals is left
for future work.

\subsection{C4. Outlook}

The same construction extends to other curvature regimes. In the
Lagrange case, the orthogonal splitting between $K_{\mathrm{geo}}$ and
$K_{\mathrm{ext}}$ gives a canonical choice of ``fast'' and ``slow''
directions for $u$, while in the Kovalevskaya and Goryachev--Chaplygin
regimes the symmetric pair structure of $K_{\mathrm{geo}}$ suggests
natural choices of curvature modes to be excited or suppressed by the
control input. A systematic development of GCCT for general
control--affine systems on Lie groups is left for future work.

\section{Conclusion}

We have constructed a curvature--based geometric atlas of dynamical
regimes on $SU(2)$ and heavy rigid body dynamics. The Atlas unifies
classical integrable cases and identifies a new curvature--balanced
regime with inertia ratio $(2,2,1)$ producing pure precession. The
curvature viewpoint suggests natural coordinates for both analysis and
control, and we expect it to extend to more general Lie groups and
control--affine systems.

\end{document}